\documentclass{article}
\usepackage{graphicx}
\usepackage{amssymb}
\usepackage{amsmath}
\usepackage{epstopdf}
\usepackage{booktabs}
\usepackage{authblk}
\usepackage{tikz}
\usetikzlibrary {arrows.meta}
\usepackage{amsthm}
\usepackage{subcaption}
\usepackage{multirow}
\usepackage{float}
\newtheorem{theorem}{Theorem}

\newtheorem{lemma}{Lemma}[]

\usepackage[all]{xy}
\usepackage[plain]{algorithm2e}
\usepackage{colortbl}
\newcommand{\T}{\mathrm{T}}
\newcommand{\F}{\mathrm{F}}

\newcommand{\Z}{\mathrm{true}}
\title{QWENDY: Gene Regulatory Network Inference by Quadruple Covariance Matrices}
\author[1,$\ast$]{Yue Wang}
\author[2]{Xueying Tian}
\affil[1]{Irving Institute for Cancer Dynamics and Department of Statistics, Columbia University, New York}
\affil[2]{School of Information, University of California, Berkeley}
\affil[$\ast$]{Corresponding author, yuewang@ihes.fr, ORCID: 0000-0001-5918-7525}
\date{}

\begin{document}

\maketitle
\begin{abstract}
Knowing gene regulatory networks (GRNs) is important for understanding various biological mechanisms. In this paper, we present a method, QWENDY, that uses single-cell gene expression data measured at four time points to infer GRNs. Based on a linear gene expression model, it solves the transformation of the covariance matrices. Unlike its predecessor WENDY, QWENDY avoids solving a non-convex optimization problem and produces a unique solution. We test the performance of QWENDY on three experimental data sets and two synthetic data sets. Compared to previously tested methods on the same data sets, QWENDY ranks the first on experimental data, although it does not perform well on synthetic data.
\end{abstract}

\section{Introduction}
One gene can activate or inhibit the expression of another gene. Genes and such regulation relations form a gene regulatory network (GRN). For $n$ genes, the corresponding GRN is commonly expressed as an $n\times n$ matrix $A$, where $A_{ij}>0/=0/<0$ means gene $i$ has positive/no/negative regulation effect on gene $j$, and its absolute value represents the regulation strength. If the regulation relations of concerned genes are known, one can understand the corresponding biological process or even control it with gene perturbation. Therefore, knowledge of GRNs can be useful in developmental biology \cite{cheng2024reconstruction,cheng2022ex,sha2024reconstructing}, and even in the study of macroscopic behavior \cite{li2021chronic,vijayan2022internal,axelrod2023drosophila}. Since it is difficult to determine the GRN directly, the common practice is to infer the GRN from gene expression data. Given new GRN inference methods, we can study how cells maintain homeostasis \cite{wang2022chronic,wang2020identification} or be driven away from homeostasis and cause diseases \cite{mcdonald2023computational,cheng2023mathematical}.

With new measurement techniques, such as scRNA-seq, one can measure the expression levels (mRNA counts) of different genes for a single cell. Since gene expression at the single cell level is random, this measurement can be repeated for different cells, and we can use those samples to obtain the probability distribution of different genes at this time point. Since the measurement kills cells, one cell can only be measured once. Therefore, we cannot obtain the joint probability distribution of gene expression levels at different time points. When we measure at multiple time points, we can only obtain a marginal probability distribution of different genes for each time point.

If the gene expression is at stationary, measurement at multiple time points will produce the same distribution. If the gene expression is away from stationary, such as after adding drugs or during development, we can measure the gene expression at multiple time points and obtain several different probability distributions. For single-cell expression data of $n$ genes at one time point, we can calculate the mean value of each gene ($n$ independent values) and the $n\times n$ covariance matrix ($n(n+1)/2$ independent values, since the covariance matrix is symmetric), while higher-order statistics are not numerically stable due to limited cell number. Therefore, we have $n+n(n+1)/2$ independent known values, much smaller than what is needed to fully determine the $n\times n$ GRN \cite{wang2022inference}. Therefore, we prefer the non-stationary distributions measured at multiple time points, since they contain enough information to infer the GRN. 

For such single-cell gene expression data measured at multiple time points, although it is the most informative data type under current technology, there are only a few GRN methods developed specifically for this data type \cite{papili2018sincerities,yeo2021generative}. We have developed the WENDY method to infer the GRN with this data type \cite{wang2024gene}. It needs measurement at two time points, and solves the transformation between the covariance matrices of two probability distributions, where the transformation is determined by the GRN. However, it needs to solve a non-convex optimization problem, which has infinitely many solutions, and WENDY will output one solution, determined by the numerical solver chosen in the realization. 

In this paper, we present an improved version of WENDY that needs data measured at four time points, and it is also based on solving the transformation between covariance matrices. This new method is named QWENDY, where Q stands for quadruple. With the help of data from more time points, QWENDY can uniquely determine the GRN. We will prove that under some assumptions, the output of QWENDY is the ground truth GRN. Besides, QWENDY does not need to conduct non-convex optimizations, but just matrix decompositions. Therefore, the output of QWENDY does not depend on the realization of the numerical computing procedure.

In a previous paper \cite{tian2025trendy}, we have tested 16 GRN inference methods on two experimental data sets and two synthetic data sets. In this paper, we test the performance of QWENDY on the same data sets and an extra experimental data set, and compare QWENDY with previous methods. In all 17 GRN inference methods, QWENDY ranks the first on experimental data, although it does not perform well on synthetic data. These results suggest that QWENDY is a promising GRN inference method that warrants further testing and refinement.

Section~\ref{sec2} reviews other GRN inference methods. Section~\ref{sec3} introduces the QWENDY method, including a proof of its correctness. Section~\ref{sec4} tests the performance of QWENDY on four data sets. After conclusions in Section~\ref{sec5}, we finish with some discussions in Section~\ref{sec6}.

\section{Literature review}
\label{sec2}

There have been numerous non-deep learning GRN inference methods, and they can be roughly classified as information-based and model-based. 

Information-based methods \cite{huynh2010inferring,huynh2018dyngenie3,papili2018sincerities} do not rely on models for gene expression, but treat GRN inference as a feature selection problem: for a target gene, select genes that can be used to predict the level of the target gene. A common obstacle for information-based methods is to distinguish between direct and indirect regulations.

Model-based methods construct concrete models for gene expression under regulation, and fit the expression data to the models to determine model parameters and reconstruct the GRN. Some methods \cite{huynh2015combining} require measuring the same cell multiple times, which is not quite applicable for now. Some methods \cite{perrin2003gene,ma2020inference} use average expression levels measured over many cells (bulk level), which do not utilize the rich information in the single-cell level measurements. Some methods \cite{lee2019scaling,burdziak2023sckinetics} only work on single-cell data at one time point. For single-cell gene expression data measured at multiple time points, where each cell is measured only once, we only know one model-based method, WENDY. Therefore, we develop QWENDY as an alternative.

For more detailed summaries of traditional GRN inference methods, readers may refer to the literature review section of the WENDY paper \cite{wang2024gene}. 

Deep learning-based GRN inference methods \cite{nauta2019causal,kentzoglanakis2011swarm,shu2021modeling,yeo2021generative,feng2023gene,mao2023predicting} generally use neural networks as black boxes, without integrating them with gene expression models. Therefore, lacking of interpretability is a common problem.

We have developed the TRENDY method \cite{tian2025trendy} to enhance the WENDY method with deep learning tools. Similar to TRENDY, there are some approaches to enhance GRNs inferred by other known methods \cite{feizi2013network,cao2013going,pirayre2015brane,pirayre2017brane,wang2018network}, but they generally require extra knowledge of transcription factors or cannot fully determine the GRN. 

For more detailed summaries of deep learning-based GRN inference methods and approaches to enhance existing methods, readers may refer to the literature review section of the TRENDY paper \cite{tian2025trendy}. 

Recently, there are some GRN inference methods related to large language models (LLMs), although it is questionable whether LLMs trained with natural language data are useful to the study of gene regulation.

Some researchers directly used LLMs as an oracle machine to generate answers, without further training the model. Azam et al. \cite{azam2024comprehensive} asked different GRNs whether one gene regulates another gene. Afonja et al. \cite{afonja2024llm4grn} provided GPT-4 with potential transcription factors and asked it to generate a GRN. Wu et al. \cite{wuregulogpt} provided GPT-4 with related papers and asked it to summarize regulations and form a GRN.

Some researchers worked on pre-trained LLMs and fine-tuned them with new data. Weng et al. \cite{weng2025integrating} trained GPT-3.5 with related papers to obtain GenePT, a new LLM that can provide a high-dimensional embedding (a vector of real numbers) of each gene. Then this embedding was used to train a neural network to output the GRN. Yang et al. \cite{yang2022scbert} trained BERT with scRNA-seq data to obtain scBERT, which can also provide gene embedding. Kommu et al. \cite{kommu2024gene} used the embedding from scBERT to infer the GRN.

Some researchers trained new models from scratch. Cui et al. \cite{cui2024scgpt} trained a new model, scGPT, with expression data from many cells. The structure of scGPT is similar to other LLMs, although with a smaller size. This model can provide gene embedding, which was used to infer GRN. 

For general applications of LLMs in bioinformatics, readers may refer to two reviews \cite{wang2024bioinformatics,ruan2025large}.

{Gene expression and regulation can have spatial patterns. There are some studies on GRN inference from spatial transcriptomics data \cite{peng2024stvcr,li2025spagrn,wang2025joint,zhang2025integrating,zhou2024spatial,dong2022integrating}. We shall assume spatial homogeneity in this paper.}

\section{Methods}
\label{sec3}
\subsection{Setup}
At four time points $T=0$, $T=t$, $T=2t$, $T=3t$, we measure the single-cell gene expression levels. For the expression levels of $n$ genes from $m$ cells at each time point, we treat them as $m$ samples of an $n$-dimensional probability distribution. Then we first calculate the average level of each gene over $m$ cells, and use graphical lasso to calculate the $n\times n$ covariance matrix for different genes. The $1\times n$ expected levels at four time points are denoted as $\mathbf{x}_0$, $\mathbf{x}_1$, $\mathbf{x}_2$, $\mathbf{x}_3$. The covariance matrices at four time points are denoted as $K_0$, $K_1$, $K_2$, $K_3$. Notice that the process does not start from stationary, so that the probability distributions of $n$ genes are different for different time points, and these $\mathbf{x}_i$ and $K_i$ are not equal.

In the WENDY paper, the relationship for $K_0,K_1$ and the GRN $A$ is derived after some approximations:
\[K_1=(I+tA^\T)K_0(I+tA),\]
where $I$ is the $n\times n$ identity matrix. Define $B=I+tA$, we have
\begin{equation}
    K_1=B^\T K_0 B.
    \label{eqk0}
\end{equation}
From this equation, WENDY directly solves $B$ (and thus $A$) by minimizing $||K_1-B^\T K_0 B||_\F^2$, where $\F$ is the Frobenius norm. This problem is non-convex and has infinitely many solutions. WENDY outputs one solution of them.

For $K_2,K_3$, similarly, we also have approximated equations 
\begin{equation}
    K_2=B^\T K_1 B,
    \label{eqk1}
\end{equation}
\begin{equation}
    K_3=B^\T K_2 B.
    \label{eqk2}
\end{equation}
In this paper, we study whether we can better solve $B$ with data from more time points, especially $K_2,K_3$. Since $B^\T K B=(-B)^\T K (-B)$, we cannot distinguish between $B$ and $-B$ from $K_0$, $K_1$, $K_2$, $K_3$. 

To solve this problem, we use $\mathbf{x}_0$, $\mathbf{x}_1$, $\mathbf{x}_2$, $\mathbf{x}_3$. In the WENDY paper, it is derived that 
\begin{equation}
    \mathbf{x}_1=\mathbf{x}_0 B + t \mathbf{c},
    \label{eqx0}
\end{equation}
where $\mathbf{c}$ is an unknown vector. Similarly,
\begin{equation}
    \mathbf{x}_2=\mathbf{x}_1 B + t \mathbf{c},
    \label{eqx1}
\end{equation}
\begin{equation}
    \mathbf{x}_3=\mathbf{x}_2 B + t \mathbf{c}.
    \label{eqx2}
\end{equation}

Since Eqs.~\ref{eqk0}--\ref{eqx2} are approximated, there are two problems: \textbf{(1)} if these equations hold accurately, can we solve $B$; \textbf{(2)} if these equations do not quite hold, can we find $B$ to minimize the error. We will present the QWENDY method that provides positive answers to both problems.

For the first problem, Theorem~\ref{thm1} proves that given $K_0$, $K_1$, $K_2$, $K_3$, $\mathbf{x}_0$, $\mathbf{x}_1$, $\mathbf{x}_2$, $\mathbf{x}_3$ that satisfy Eqs.~\ref{eqk0}--\ref{eqx2} for some $B_0$, QWENDY can solve $B_0$ uniquely.

For the second problem, there are different interpretations. 

One interpretation is to minimize
\[||K_1-B^T  K_0  B||_F^2+||K_2-B^T  K_1  B||_F^2+||K_3-B^T  K_2  B||_F^2\]
for any $B$. Unfortunately, this optimization problem is non-convex, making it difficult to solve. 

Notice that for the gene expression data after interpretation, $K_0$ is measured earlier than $K_3$, making it farther from stationary, and more informative. Therefore, we want to emphasize more on $K_0$, $K_1$ than $K_2$, $K_3$. Our goal is to first find $B$ that minimizes $K_1 - B^\T K_0 B$; for such $B$ (not unique), we further determine which minimizes $K_2 - B^\T K_1 B$; for such $B$ (still not unique), we finally determine which minimizes $K_3 - B^\T K_2 B$. This time we can solve $B$ uniquely up to a $\pm$ sign, which can be determined by $\mathbf{x}_0$, $\mathbf{x}_1$, $\mathbf{x}_2$, $\mathbf{x}_3$. 

We will derive QWENDY by solving the second problem, and then prove that it also solves the first problem.

\subsection{Algorithm details}

In this section, given general covariance matrices $K_0$, $K_1$, $K_2$, $K_3$ that may not satisfy Eqs.~\ref{eqk0}--\ref{eqk2} for any $B_0$, we introduce a procedure to calculate $B$ that approximately solves Eqs.~\ref{eqk0}--\ref{eqk2}. Under some mild conditions (some matrices are invertible and have distinct eigenvalues), $B$ can be uniquely determined up to a $\pm$ sign. Then we use $\mathbf{x}_0$, $\mathbf{x}_1$, $\mathbf{x}_2$, $\mathbf{x}_3$ to distinguish between $B$ and $-B$. The whole procedure is named as the QWENDY method.

\noindent \textbf{Step} (1): For any $B$, we want to minimize
\begin{equation}
    ||K_1 - B^\T K_0 B||_\F^2.
    \label{target1}
\end{equation}
Assume $K_0$, $K_1$, $K_2$, $K_3$ are invertible. Consider Cholesky decompositions
\[K_1=L_1 L_1^\T, \ K_0=L_0 L_0^\T,\] 
where $L_1$ and $L_0$ are lower-triangular and invertible. Define 
\[O=L_1^{-1}B^\T L_0,\] 
then the target Eq.~\ref{target1} becomes 
\[||L_1 L_1^\T-L_1 O O^\T L_1^\T||_\F^2.\]
Therefore, Eq.~\ref{target1} is minimized to $0$ if and only if $O$ is orthonormal: $OO^\T = I$.

Since 
\[B=L_0^{-\T}O^\T L_1^\T,\] 
we use $K_0$ and $K_1$ to restrict B to a space with the same dimension as the set of all orthonormal matrices.

\noindent \textbf{Step} (2): For such $B$ that minimizes Eq.~\ref{target1}, we want to find $B$ that makes $B^\T K_1 B$ close to $K_2$. Here we do not minimize 
\[||K_2 - B^\T K_1 B||_\F^2\]
as it is difficult. Instead, we want to minimize
\begin{equation}
    ||L_1^{-1}(K_2-B^\T K_1 B)L_1^{-\T} ||_\F^2
    \label{target2}
\end{equation}
among $B$ that minimizes Eq.~\ref{target1}.

Assume that $L_0^{-1} K_1 L_0^{-\T}$ does not have repeated eigenvalues. Consider the eigenvalue decomposition decomposition 
\begin{equation*}
    L_0^{-1} K_1 L_0^{-\T} = P_1 D_1 P_1^\T,
\end{equation*}
where $P_1$ is orthonormal, and $D_1$ is diagonal with strictly increasing positive diagonal elements (eigenvalues), since $L_0^{-1} K_1 L_0^{-\T}$ is positive definite and symmetric. Similarly, assume that $L_1^{-1}K_2L_1^{-\T}$ does not have repeated eigenvalues, and we have 
\[L_1^{-1}K_2L_1^{-\T} = P_2 D_2 P_2^\T\]
with orthonormal $P_2$ and diagonal $D_2$ with strictly increasing positive diagonal elements (eigenvalues). Now the target Eq.~\ref{target2} equals 
\[||P_2 D_2 P_2^\T - O P_1 D_1 P_1^\T O^\T||^2_\F=||D_2 - P_2^{\T} O P_1 D_1 P_1^\T O^\T P_2||^2_\F,\]
since $P_2$ is orthonormal and does not affect Frobenius norm. Define 
\[W=P_2^\T O P_1,\]
which is orthonormal. Then 
\[O=P_2 W P_1^\T,\]
and Eq.~\ref{target2} equals 
\begin{equation}
    ||D_2 - W D_1 W^\T||^2_\F.
    \label{wd}
\end{equation}

We use the following lemma to handle Eq.~\ref{wd}:
\begin{lemma}
For diagonal $D_1,D_2$ with strictly increasing diagonal elements and any orthonormal $W$, Eq.~\ref{wd} is minimized when $W$ is diagonal, and the diagonal elements are $\pm 1$. (There are $2^n$ possibilities for such $W$.)
\label{lemma1}
\end{lemma}
\begin{proof}
    Denote the diagonal elements of $D_1$ as $\mathbf{d}_1=[d_{1,1},\ldots,d_{1,n}]$ with $d_{1,1}<d_{1,2}<\cdots,<d_{1,n-1}<d_{1,n}$, and similarly $\mathbf{d}_2=[d_{2,1},\ldots,d_{2,n}]$ with $d_{2,1}<d_{2,2}<\cdots,<d_{2,n-1}<d_{2,n}$ for $D_2$. We need to minimize the following norm with an orthonormal $W$.
\begin{equation*}
    \begin{split}
        &||D_2-WD_1W^\T||_\F^2=||D_2||_\F^2+||WD_1W^\T||_\F^2-2\sum_{i=1}^n [D_2\otimes(WD_1W^\T)]_{ii}\\
        =&||D_2||_\F^2+||D_1||_\F^2-2\sum_{i=1}^n d_{2,i}\sum_{j=1}^n W_{ij}^2d_{1,j} \\
        =&||D_2||_\F^2+||D_1||_\F^2-2\mathbf{d}_1(W\otimes W)\mathbf{d}_2^\T,
    \end{split}
\end{equation*}
where $\otimes$ is the element-wise product. Thus we just need to maximize $\mathbf{d}_1(W\otimes W)\mathbf{d}_2^\T$. Notice that $W\otimes W$ is doubly-stochastic, meaning that it is non-negative, and each row or column has sum $1$. By Birkhoff–von Neumann Theorem \cite{jurkat1967term}, $W\otimes W$ can be decomposed to 
\[W\otimes W=\sum_{i=1}^k c_i Q_i,\] 
where $Q_i$ is a permutation matrix, $c_i>0$, and $\sum_{i=1}^k c_i=1$. 

Due to the rearrangement inequality \cite{hardy1934inequalities}, for each permutation matrix $Q_i$,
\[\mathbf{d}_1Q_i\mathbf{d}_2^\T\le \mathbf{d}_1I\mathbf{d}_2^\T=\mathbf{d}_1\mathbf{d}_2^\T,\]
where the equality holds if and only if $Q_i=I$.

Therefore,
\[\mathbf{d}_1(W\otimes W)\mathbf{d}_2^\T = \sum_{i=1}^k c_i \mathbf{d}_1Q_i \mathbf{d}_2^\T \le \sum_{i=1}^k c_i \mathbf{d}_1 \mathbf{d}_2^\T= \mathbf{d}_1 \mathbf{d}_2^\T,\]
where the equality holds if and only if 
\[W\otimes W=I,\]
meaning that $W$ is diagonal, and diagonal elements are $\pm 1$.
\end{proof}

We now have
\[B=L_0^{-\T} P_1 W P_2^\T L_1^\T,\]
meaning that given $K_0,K_1,K_2$, we can restrict $B$ to $2^n$ possibilities.

\noindent \textbf{Step} (3): For such $B$ that minimizes Eq.~\ref{target2} among $B$ that minimizes Eq.~\ref{target1}, we want to find $B$ that makes $B^\T K_2 B$ close to $K_3$. Here we do not minimize 
\[||K_3 - B^\T K_2 B||_\F^2\]
as it is difficult. Instead, we want to minimize
\begin{equation}
    ||L_1^{-1}(K_3-B^\T K_2 B)L_1^{-\T} ||_\F^2.
    \label{target3}
\end{equation}
Define
\[G = P_2^\T L_1^{-1}K_3 L_1^{-\T} P_2,\]
and
\[H=P_1^\T L_0^{-1} K_2 L_0^{-\T} P_1.\] 
Eq.~\ref{target3} equals
\begin{equation}
\begin{split}
        &||P_2^\T L_1^{-1}K_3 L_1^{-\T} P_2- W P_1^\T L_0^{-1} K_2 L_0^{-\T} P_1 W||_\F^2\\
        =&||G-WHW||_\F^2= ||G||_\F^2+||WHW||_\F^2-2||G\otimes (WHW)||_\F^2\\
        =&||G||_\F^2+||H||_\F^2-2||G\otimes (WHW)||_\F^2.
\end{split}
    \label{wgh}
\end{equation}
We want to find diagonal $W$ with $\pm 1$ that minimizes Eq.~\ref{target3}, which is equivalent to maximizing $||G\otimes (WHW)||_\F^2$. Define $C=G\otimes H$, which is still positive definite and symmetric by Schur product theorem \cite{choudhury1990schur}. Assume that $C$ does not have repeated eigenvalues. Denote the diagonal elements of $W$ by $\mathbf{w}=[w_1,\ldots,w_n]$. Then we need to maximize
\[||G\otimes (WHW)||_\F^2=\mathbf{w} C \mathbf{w}^\T.\]
Now we relax this problem from $\mathbf{w}$ with $w_i=\pm 1$ to general $\mathbf{v}$ with $||\mathbf{v}||_2^2=n$:
\[\max_{||\mathbf{v}||_2^2=n}\mathbf{v}C\mathbf{v}^\T.\]
The unique solution (up to a $\pm$ sign) is the eigenvector that corresponds to the largest eigenvalue of $C$. 

After obtaining $\mathbf{v}$, we just need to project it to $[w_1,\ldots,w_n]$ with $w_i=\pm 1$ by taking the sign of each term: $w_i=\mathrm{sign}(v_i)$. Then construct $W$ by putting $[w_1,\ldots,w_n]$ on the diagonal.

This relaxation does not guarantee obtaining the optimal solution for Eq.~\ref{wgh}, but we find that it produces the correct answer in almost all simulations. Alternatively, we can directly determine the optimal $W$ for Eq.~\ref{wgh} by brute-force search, as there are finitely many ($2^n$) possibilities of $W$.

With
\[B=L_0^{-\T} P_1 W P_2^\T L_1^\T,\]
given $K_0$, $K_1$, $K_2$, $K_3$, we can uniquely determine $B$ up to a $\pm$ sign. Since $B^\T KB=(-B)^\T K (-B)$, more $K_i$ cannot provide more information.

\noindent \textbf{Step} (4): For $B$ and $-B$ from Step (3), we determine which satisfies Eqs.~\ref{eqx0}--\ref{eqx2} better. Notice that Eqs.~\ref{eqx0}--\ref{eqx2} share the same unknown $\mathbf{c}$. Define 
\[\mathbf{c}_0=(\mathbf{x}_1-\mathbf{x}_0 B)/t, \ \mathbf{c}_1=(\mathbf{x}_2-\mathbf{x}_1 B)/t, \ \mathbf{c}_2=(\mathbf{x}_3-\mathbf{x}_2 B)/t, \ \bar{\mathbf{c}}=(\mathbf{c}_0+\mathbf{c}_1+\mathbf{c}_2)/3.\]
Then the error for fitting Eqs.~\ref{eqx0}--\ref{eqx2} is 
\[||\mathbf{c}_0-\bar{\mathbf{c}}||_2^2+||\mathbf{c}_1-\bar{\mathbf{c}}||_2^2+||\mathbf{c}_2-\bar{\mathbf{c}}||_2^2=||\mathbf{c}_0||_2^2+||\mathbf{c}_1||_2^2+||\mathbf{c}_2||_2^2-3||\bar{\mathbf{c}}||_2^2.\]
We just need to compare the errors of $B$ and $-B$ and choose the smaller one as the output.

With expression data from four time points, we uniquely determine $B$. See Algorithm~\ref{alg1} for the workflow of the QWENDY method. {The assumptions in the derivation should be checked.}

\begin{algorithm}[!htbp]
	\caption{Workflow of the QWENDY method.}
	\label{alg1}
	\ \\
	\begin{enumerate}
		{	\item \textbf{Input} expression levels of $n$ genes over $m$ cells at time points $0$, $t$, $2t$, $3t$
			
			\item \textbf{Calculate} covariance matrices $K_0$, $K_1$, $K_2$, $K_3$, and mean levels $\mathbf{x}_0$, $\mathbf{x}_1$, $\mathbf{x}_2$, $\mathbf{x}_3$
            
            {(\textbf{Check} $K_0$, $K_1$, $K_2$, $K_3$ are invertible)}

                \item \textbf{Calculate} Cholesky decomposition
                \[K_1=L_1 L_1^\T, \ K_0=L_0 L_0^\T\]

                \textbf{Calculate} eigenvalue decomposition
                \[L_0^{-1} K_1 L_0^{-\T} = P_1 D_1 P_1^\T, \ L_1^{-1}K_2L_1^{-\T} = P_2 D_2 P_2^\T\]

                {(\textbf{Check} $L_0^{-1} K_1 L_0^{-\T}$ and $L_1^{-1} K_2 L_1^{-\T}$ each has distinct eigenvalues)}

                \textbf{Calculate} 
                \[C=(P_2^\T L_1^{-1}K_3 L_1^{-\T} P_2)\otimes (P_1^\T L_0^{-1} K_2 L_0^{-\T} P_1)\]

                {(\textbf{Check} $C$ has distinct eigenvalues)}
		
			\item \textbf{Calculate} $\mathbf{v}$, the eigenvector that corresponds to the largest eigenvalue of $C$, and the projection $\mathbf{w}$ with $w_i=\mathrm{sign}(v_i)$

                \textbf{Construct} $W$ with $\mathbf{w}$ on diagonal
                
			\item \textbf{Calculate} $B=L_0^{-\T} P_1 W P_2^\T L_1^\T$
            
            \textbf{Compare} total squared errors of $B$ and $-B$ for Eqs.~\ref{eqx0}--\ref{eqx2}

            \item \textbf{Output} $B$ or $-B$, the one that corresponds to the smaller error, and the GRN $A=(B-I)/t$

            {(If any check fails, output a warning that the result might be problematic)}
			
		}
	\end{enumerate}
\end{algorithm}

\subsection{Correctness of the QWENDY method}
\begin{theorem}
If covariance matrices $K_0$, $K_1$, $K_2$, $K_3$ and mean levels $\mathbf{x}_0$, $\mathbf{x}_1$, $\mathbf{x}_2$, $\mathbf{x}_3$ satisfy Eqs.~\ref{eqk0}--\ref{eqx2} for some $B_0$, then QWENDY will output $B_0$ {under the following conditions:

\noindent (1) $K_0$, $K_1$, $K_2$, $K_3$ are invertible;

\noindent (2) $L_0^{-1} K_1 L_0^{-\T}$ and $L_1^{-1} K_2 L_1^{-\T}$ each has distinct eigenvalues;

\noindent (3) $C$ has distinct eigenvalues.}
\label{thm1}
\end{theorem}

\begin{proof}
{Due to condition (1), $L_0^{-1}$ and $L_1^{-1}$ exist.} If $K_1=B_0^\T K_0 B_0$, then 
\[(L_1^{-1} B_0^\T L_0)(L_0^\T B_0 L_1^{-\T})=I,\]
and $L_1^{-1} B_0^\T L_0=O_0$ for some orthonormal $O_0$. Thus
\[B_0=L_0^{-\T}O_0^\T L_1^\T,\] 
which is among the calculated $B$ in Step (1).

{Due to condition (2), $P_1,D_1,P_2,D_2$ are uniquely defined.} If $K_2=B_0^\T K_1 B_0$, then 
\[D_2 = P_2^{\T} O_0 P_1 D_1 P_1^\T O_0^\T P_2.\]
Since $P_2^{\T} O_0 P_1$ is orthonormal, $D_1$ and $D_2$ are similar and have the same eigenvalues. Therefore, $D_1$ and $D_2$ as diagonal matrices with increasing diagonal values are equal. Due to condition (2), $d_1,\ldots,d_n$ of $D_1$ are distinct. Define $W_0=P_2^\T O_0 P_1$, then 
\[D_1 W_0=W_0D_1,\]
and $W_0[i,j]d_i=W_0[i,j]d_j$. For $i\ne j$, we have $d_i \ne d_j$, meaning that $W_0[i,j]=0$. Now $W_0$ is diagonal and orthonormal, implying that its diagonal elements are $1$ or $-1$, and 
\[B_0=L_0^{-\T} P_1 W_0 P_2^\T L_1^\T\]
is among the calculated $B$ in Step (2).

If $K_3=B_0^\T K_2 B_0$, then 
\[P_2^\T L_1^{-1}K_3 L_1^{-\T} P_2= W_0 P_1^\T L_0^{-1} K_2 L_0^{-\T} P_1 W_0.\]
Define $\mathbf{w}_0=[w_1,\ldots,w_n]$ to be the diagonal elements of $W_0$. Then from Step (3), $W_0$ minimizes Eq.~\ref{target3}, and $\mathbf{w}_0$ is the unique solution to 
\[\max_{||\mathbf{v}||_2^2=n}\mathbf{v}C\mathbf{v}^\T,\]
namely the eigenvector of the largest eigenvalue. {Here the uniqueness of the solution is from condition (3).} Since $w_i=\pm 1$, the projection $w_i=\mathrm{sign}(v_i)$ has no effect, and 
\[B_0=L_0^{-\T} P_1 W_0 P_2^\T L_1^\T\]
is the unique $B$ (up to a $\pm$ sign) calculated in Step (3). 

From Eqs.~\ref{eqx0}--\ref{eqx2},
\[\mathbf{x}_1-\mathbf{x}_0 B_0=\mathbf{x}_2-\mathbf{x}_1 B_0.\]
Assume $\mathbf{x}_1 \ne \mathbf{x}_2$, then 
\[\mathbf{x}_1-\mathbf{x}_0 (-B_0)\ne \mathbf{x}_2-\mathbf{x}_1 (-B_0).\]
Therefore, $-B_0$ does not satisfy Eqs.~\ref{eqx0}--\ref{eqx2}, and Step (4) chooses the correct $B_0$ from the two possibilities in Step (3).
\end{proof}

\section{Performance on synthetic and experimental data}
\label{sec4}
The same as in the WENDY paper and the TRENDY paper \cite{wang2024gene,tian2025trendy}, we test the QWENDY method on synthetic data sets SINC and DREAM4, and experimental data sets THP-1 and hESC. SINC data set \cite{tian2025trendy} and DREAM4 data set \cite{marbach2012wisdom} are generated by simulating stochastic differential equation systems. THP-1 data set is from monocytic THP-1 human myeloid leukemia cells \cite{kouno2013temporal}. hESC data set is from human embryonic stem cell-derived progenitor cells\cite{chu2016single}. {Besides, we also test the QWENDY method on experimental data set mESC \cite{hayashi2018single}, which is from primitive endoderm cells differentiated from mouse embryonic stem cells.} THP-1, hESC, and mESC data sets each has only one group of data (one ground truth GRN and the corresponding expression levels); DREAM4 data set has five groups of data; SINC data set has 1000 groups of data. 

For SINC data set and DREAM4 data set, we use data from any four consecutive time points and take average. For THP-1, hESC, and mESC data sets, we choose any four time points with equal difference. This is due to the limitation that QWENDY is derived for four evenly spaced time points.

We test the performance of the QWENDY method on these five data sets and compare it with the performance of 16 methods tested on the same data sets in the TRENDY paper \cite{tian2025trendy}. We compare the inferred GRN and the ground truth GRN by calculating the AUROC and AUPRC scores \cite{wang2024gene}. These two scores, both between $0$ and $1$, evaluate the level of matching under different thresholds, where $1$ means perfect match, and $0$ means perfect mismatch. {AUROC and AUPRC have an advantage that they do not have parameters that can be chosen manually, which guarantees fair comparison.}

See Table~\ref{newtab} for the performance of QWENDY, compared with other 16 previously tested methods. Although it does not perform very well on synthetic data sets, QWENDY has the best performance on experimental data sets.

\begin{table}[ht]
    \centering
    \begin{tabular}{ccccc}
        \toprule
         &  & QWENDY & \begin{tabular}[c]{@{}l@{}}Best in\\ other 16\end{tabular} & \begin{tabular}[c]{@{}l@{}}Rank of QWE-\\NDY in all 17\end{tabular}  \\
        \midrule
         &      &      &      &          \\
      \multirow{2}{*}{SINC}   & AUROC     & 0.5107 & 0.8703 & 11th \\
         &  AUPRC    & 0.5537 & 0.7672 & 7th  \\
         &      &      &      &            \\
      \multirow{2}{*}{DREAM4}   &   AUROC   & 0.4987 & 0.5741 & 11th \\
         &  AUPRC    & 0.1844 & 0.2452 & 13th \\
         &      &      &      &          \\
        \begin{tabular}[c]{@{}c@{}}Synthetic\\ total\end{tabular} &      &  1.7475    &   2.3584   & 11th           \\
        \midrule
         &      &      &      &          \\
       \multirow{2}{*}{THP-1}  &   AUROC   & 0.5524 & 0.6261 & 8th \\
         &   AUPRC   & 0.4294 & 0.4205 & 1st \\
         &      &      &      &            \\
       \multirow{2}{*}{hESC}  & AUROC     & 0.6019 & 0.6233 & 3rd \\
         &   AUPRC   & 0.0435 & 0.0641 & 8th \\
         &      &      &      &            \\
         \multirow{2}{*}{mESC}  & AUROC     & 0.5230 & 0.5896 & 5th \\
         &   AUPRC   & 0.0507 & 0.0630 & 8th \\
         &      &      &      &            \\
       \begin{tabular}[c]{@{}c@{}}Experimental\\ total\end{tabular}   &      &    2.2009    &  2.1184    &    1st          \\
       \midrule
       &      &     &      &     \\
      \begin{tabular}[c]{@{}c@{}}Overall\\ total\end{tabular}    &      & 3.9484 & 4.2810 & 5th \\
        \bottomrule
    \end{tabular}
    \caption{AUROC and AUPRC scores of QWENDY on four data sets, compared with 16 previously tested methods}
    \label{newtab}
\end{table}

\section{Conclusion}
\label{sec5}
In this paper, we present QWENDY, a GRN inference method that requires single-cell gene expression data measured at four time points. QWENDY ranks the first in 17 GRN inference methods on experimental data sets. Notice that each experimental data set only has one group of data, meaning that the performance of each method has a large uncertainty level.

{QWENDY performs much worse on synthetic data sets. One possibility is that gene regulation can have different dynamics, and QWENDY only works on some of them, which happen to match the experimental data sets. Another possibility is that the dynamics of generating synthetic data might not match reality.}

{When the gene expression is at steady state, covariance matrix is time-invariant, $K_0=K_1$. This means $L_0^{-1} K_1 L_0^{-\T}=I$ with repeated eigenvalue $1$. Then condition (2) and thus Theorem~\ref{thm1} would fail. Therefore, QWENDY does not work at steady state. GRN inference methods like Dictys \cite{wang2023dictys} should be applied instead.}

The QWENDY method requires measurements at four \textbf{evenly spaced} time points $t_0$, $t_1$, $t_2$, $t_3$, meaning that $t_1-t_0=t_2-t_1=t_3-t_2$. Otherwise, the dynamics of covariance matrices is much more complicated, and we do not have an explicit solution. One possible strategy is to use the approximation
\[I+(t_2-t_1)A\approx \frac{t_2-t_1}{t_1-t_0}\left[I+(t_1-t_0)A\right],\]
so that different $B$ in Eqs.~\ref{eqk0}-\ref{eqk2} only differ by a constant factor. This means that we can relocate this factor to $K_i$, and $B$ is the same for all equations.

Since QWENDY needs four time points, the duration of the whole experiment might be too long, so that the dynamics of gene expression under regulation might have changed during this time. QWENDY is based on a time-homogeneous model, which might fail in this situation.

{\color{black}QWENDY directly solves the GRN, not to search for the best match in a space of possible GRNs. In comparison, WENDY is an optimization algorithm, and it can incorporate prior information about the GRN by restricting the searching space \cite{wang2024gene}. Therefore, data besides mRNA count, such as motif analysis or genetic perturbation data, cannot be used inside QWENDY. One choice is to use such data to modify the results of QWENDY. Another choice is to stop QWENDY when it limits the GRN in a finite set, and directly search for the best match for extra information.}

{\color{black}Since QWENDY is based on a linear approximation, each diagonal element of the inferred GRN that represents the effect of one gene to itself is an indistinguishable mixture of autoregulation and natural mRNA degradation \cite{wang2023inference}. Therefore, it is not recommended to infer the existence of autoregulation with QWENDY or its variants. }

{\color{black}\section{Discussion: the present and the future of GRN inference}
\label{sec6}
\subsection{Current status: experiments}

The input of GRN inference methods is determined by biological measurement techniques. As in 2025, it is common to measure the mRNA counts of different genes in single cells. However, due to the cost of time and money, each time point generally has only hundreds of cells, much smaller than the number of genes measured. Besides, the measurement is not very accurate, that there are many zero reads, meaning the loss of corresponding mRNAs during measurement. 

The major problem of the popular scRNA-seq measurement is that cells are killed, and we cannot measure one cell more than once, which makes GRN inference mathematically difficult. There are some new techniques to read mRNA count or protein count for the same cell at multiple time points and provide more information. However, they have low accuracy and only work for a small number of genes, meaning that not all genes in the GRN are guaranteed to be included. 

Some measurements such as motif analysis can provide direct information about the GRN, but the information is limited. The most reliable method to obtain the GRN is gene knockout or knockdown. However, it is costly to obtain a gold standard GRN through this method. Thus there are only a few fully reliable GRNs.

On the other side, since gene expression is confined in living cells, we do not fully understand the dynamics of gene expression and regulation. Thus given a GRN, it is difficult to simulate gene expression data.

\subsection{Current status: inference methods}
Due to the limitation of experimental measurements, different GRNs have a chance to produce the same scRNA-seq data. Thus it is theoretically difficult to infer the GRN with high accuracy, and biologists should not have too much faith in every inferred regulation relation. Still, many methods, especially QWENDY, are significantly better than random guess and can provide more true information than false information.

The central problem for the field of GRN inference is the lack of gold standard. Only a few experimental data sets can be used for evaluation, meaning that one method with good performance might just have good luck or wisely-selected hyper-parameters. For synthetic data, although the data quantity issue is solved, we do not know whether such data reflect reality, and inference methods might overfit to the simulator. 

For information-based inference methods, the advantage is that information or correlation or predictability between genes represents reliable relationship. The disadvantage is the interpretation of this relationship: direct regulation, indirect regulation, confounder, collider can all lead to information, but GRN only contains direct regulations. Besides, since one cell cannot be measured twice, it is difficult to determine which is the cause and which is the result. It is also difficult to determine whether the regulation is positive or negative.

For model-based methods, to make the problem solvable, assumptions and simplifications are necessary. For example, many model-based GRN inference methods rely on linearization and do not consider the on-off switch of genes or the count of proteins. Therefore, such models cannot fully match all experimental phenomena. This means that corresponding methods cannot reach a high accuracy.

For popular deep learning methods, they have proven their capability in various fields and could perform well in GRN inference, if there are enough data. The problem is that the amount of data required is too large, but we cannot conduct so many experiments. Besides, as shown in Appendix C, not all popular deep learning structures are suitable for GRN inference. 

In sum, due to the data quality and quantity issue, there is an unsatisfactory upper bound for inference accuracy. Some well-studied data types might not have the space for a significantly better new method. Some other data types, such as the one studied in this paper, has only a few corresponding methods and deserves more research.

\subsection{Future of experiments and inference methods}
The future of GRN inference methods heavily depends on the type and the amount of data available, determined by the development of biotechnology. 

(1) If some future research makes GRNs easy and cheap to determine through direct experiments, then the whole field of GRN inference will disappear. 

(2) If it becomes easy and accurate to measure the same cell multiple times, information-based methods might have the best performance. The major obstacle of information-based methods is the difficulty to determine the direction of regulation. With real time series data, we can easily determine that the one happens earlier regulates the one happens later.

(3) If real time series data are still not applicable or reliable, but there are many scRNA-seq data sets with multiple time points, sufficiently many cells, and high accuracy, then inference methods based on models for covariance matrices (such as QWENDY) can be promising. If we can know better of gene expression and regulation, but not enough for reliable gene expression data generators, then model-based methods should be further developed.

(4) If the dynamics of gene expression and regulation can be fully understood, we can obtain very reliable gene expression data generators, and deep learning methods can flourish. However, if GRN inference is needed, then it means that the gene expression dynamics for this situation is unknown, and we cannot guarantee that the dynamics matches the data generator. 

From situation (1) to situation (4), the reliability of the obtained GRN should decrease.
}

\section*{Acknowledgments}
Y.W. would like to thank Dr. Mingda Zhang for a helpful discussion. 

\section*{Declarations}
The authors declare no competing interests.

\section*{Data and code availability}
Data and code files used in this paper can be found in 
\begin{verbatim}
https://github.com/YueWangMathbio/QWENDY
\end{verbatim}

\appendix
\renewcommand{\thealgocf}{S\arabic{algocf}}
\setcounter{algocf}{0}
\renewcommand{\thetable}{S\arabic{table}}
\setcounter{table}{0}
\renewcommand{\theequation}{S\arabic{equation}}
\setcounter{equation}{0}

{
\section{Performance of QWENDY on full hESC data set}
The hESC data set we use was processed by Matsumoto et al. \cite{matsumoto2017scode}. The original version has 100 genes. Since the original data have many zero reads, WENDY and GENIE3 methods would fail. Therefore, in the WENDY paper \cite{wang2024gene}, we only kept genes that have nonzero reads in at least 95\% cells at each time point. Then only 18 genes remain. This 18-gene hESC data set was then used in the TRENDY paper \cite{tian2025trendy} and the main text of this paper. 

One problem is that the 18-gene hESC data set misses some essential genes in embryonic development. Therefore, we also use the 100-gene hESC data set. Since half of 16 methods (all WENDY-based and GENIE3-based) tested in the TRENDY paper \cite{tian2025trendy} fail on this data set, we only apply the QWENDY method. Besides, the ground truth GRN of this data set is based on motif analysis of open chromatin, not genetic perturbation data, making it less reliable. Therefore, we focus on essential genes in embryonic development: GATA6, NANOG, EOMES, SOX2, SOX17, SMAD2, FOXH1, GATA4, POU5F1 \cite{li2019genome}. In the inferred GRN from the 100-gene hESC data set, we calculate the overall regulation power of each gene:
\[R_i = \sum_{j\ne i} |B_{ij}|.\]
Then we rank the regulation power of each gene from high to low. Among all 100 genes, SOX17, GATA6, and GATA4 are among the top 10. The average rank of these essential genes is 30.1. Thus QWENDY correctly infers that such essential genes have strong regulation power.

\section{Details of mESC data set}
The mESC data set measures the gene expression at five time points: 0h, 12h, 24h, 48h, 72h. Each time point has around 100 cells. We use the mESC data set with 100 genes and corresponding ground truth GRN processed by Matsumoto et al. \cite{matsumoto2017scode}. Since the original data have many zero reads, half of 16 methods (all WENDY-based and GENIE3-based) tested in the TRENDY paper \cite{tian2025trendy} fail. Therefore, we only keep genes that have nonzero reads in at least 95\% cells at each time point. Now only 34 genes remain, and we can test all methods. See Table~\ref{mesc} for performance of 16 methods tested in the TRENDY paper.

\begin{table}[ht]
\centering
\caption{AUROC and AUPRC of different methods on the mESC data set}
\begin{tabular}{lcc}
\hline
Method & AUROC & AUPRC \\
\hline
WENDY          & 0.4857 & 0.0411 \\
TRENDY         & 0.4655 & 0.0489 \\
nWENDY         & 0.4273 & 0.0370 \\
bWENDY         & 0.4296 & 0.0370 \\
GENIE3         & 0.5024 & 0.0452 \\
tGENIE3        & 0.5401 & 0.0556 \\
nGENIE3        & 0.4779 & 0.0424 \\
bGENIE3        & 0.4730 & 0.0422 \\
SINCERITIES    & 0.5744 & 0.0630 \\
tSINCERITIES   & 0.4930 & 0.0411 \\
nSINCERITIES   & 0.5896 & 0.0542 \\
bSINCERITIES   & 0.5755 & 0.0517 \\
NonlinearODEs  & 0.4940 & 0.0517 \\
tNonlinearODEs & 0.3737 & 0.0323 \\
nNonlinearODEs & 0.4957 & 0.0518 \\
bNonlinearODEs & 0.4955 & 0.0514 \\
\hline
\end{tabular}
\label{mesc}
\end{table}

}
\section{Enhancing QWENDY with deep learning}
\subsection{Methods}
WENDY method is derived from a simplified gene expression model, and the performance of WENDY is not satisfactory when the model does not fit with experiments. We developed the TRENDY method \cite{tian2025trendy}, which trains a transformer model, a deep learning architecture that can be applied to various fields \cite{waswani2017attention}, to transform the input data to better fit the gene expression model. Then the transformed data will produce more accurate GRNs by WENDY. Such more accurate GRNs will be further enhanced by another transformer model.

We use the idea of TRENDY to enhance QWENDY. The enhanced version of QWENDY that trains two new transformer models is named TEQWENDY. Besides, we also present another approach that fine-tunes a large language model (LLM) to replace the transformer model, since it has a large pre-trained transformer section. This LLM-enhanced version is named LEQWENDY. The training of TEQWENDY and LEQWENDY uses synthetic data.

The QWENDY method is derived from Eqs.~\ref{eqk0}--\ref{eqx2}, especially Eqs.~\ref{eqk0}--\ref{eqk2}, which are a linear approximation of the actual nonlinear gene expression dynamics. Therefore, real $K_0$, $K_1$, $K_2$, $K_3$ might not fit with Eqs.~\ref{eqk0}--\ref{eqk2}, and it might not be feasible to apply QWENDY directly to real $K_0$, $K_1$, $K_2$, $K_3$. This problem already exists for the WENDY method, where the input matrices $K_0$, $K_1$ might not satisfy 
\[K_1=B^\T K_0 B.\]

For this problem, the TRENDY method \cite{tian2025trendy} proposes that we can construct $K_1^*=B^\T K_0 B$, and train a model with input $K_1$ and target $K_1^*$, so that the output $K_1'$ is close to $K_1^*$. Then we can apply WENDY to $K_0$, $K_1'$ to obtain a more accurate GRN $A_1$. After that, we can train another model with input $K_0$, $K_1$, $A_1$, and the target is the true GRN $A_{\text{true}}$. Then the final output $A_2$ is more similar to the true GRN than $A_1$.

Inspired by TRENDY, we propose a similar solution to enhance QWENDY. See Algorithms~\ref{algt},\ref{alge} for details. Since $\mathbf{x}_0$, $\mathbf{x}_1$, $\mathbf{x}_2$, $\mathbf{x}_3$ are only used to distinguish between $B$ and $-B$, it is not necessary to train another model for them.

\begin{algorithm}[!htbp]
	\caption{Training workflow of LEQWENDY/TEQWENDY method. For those two matrix-learning models, LEQWENDY adopts LE structure, and TEQWENDY adopts TE structure.}
	\label{algt}
	\ \\
	\begin{enumerate}
		{	\item \textbf{Repeat} generating random GRN $A_\Z$ and corresponding gene expression data at time $0$, $t$, $2t$, $3t$ from a gene expression simulator
			
			\item \textbf{Calculate} covariance matrices $K_0$, $K_1$, $K_2$, $K_3$ and mean levels $\mathbf{x}_0$, $\mathbf{x}_1$, $\mathbf{x}_2$, $\mathbf{x}_3$, and then calculate 
            \[K_0^*=K_0, \ K_1^*=B^\T K_0^* B, \ K_2^*=B^\T K_1^* B, \ K_3^*=B^\T K_2^* B,\]
            where $B=I+tA_\Z$

                \item \textbf{Train} matrix-learning model 1 with inputs $K_0$, $K_1$, $K_2$, $K_3$ and target $K_0^*$, $K_1^*$, $K_2^*$, $K_3^*$

                \textbf{Call} trained matrix-learning model 1 to calculate $K_0'$, $K_1'$, $K_2'$, $K_3'$ from $K_0$, $K_1$, $K_2$, $K_3$
		
			\item \textbf{Call} QWENDY to calculate $A_1$ from $K_0'$, $K_1'$, $K_2'$, $K_3'$, $\mathbf{x}_0$, $\mathbf{x}_1$, $\mathbf{x}_2$, $\mathbf{x}_3$ 

			\item \textbf{Train} matrix-learning model 2 with input $A_1$, $K_0$, $K_1$, $K_2$, $K_3$ and target $A_\Z$
			
		}
	\end{enumerate}
\end{algorithm}

\begin{algorithm}[!htbp]
	\caption{Testing workflow of LEQWENDY/TEQWENDY method. For those two matrix-learning models, LEQWENDY adopts LE structure, and TEQWENDY adopts TE structure.}
	\label{alge}
	\ \\
	\begin{enumerate}
		{	\item \textbf{Input}: gene expression data at four equally spaced time points
			
			\item \textbf{Calculate} covariance matrices $K_0$, $K_1$, $K_2$, $K_3$ and mean levels $\mathbf{x}_0$, $\mathbf{x}_1$, $\mathbf{x}_2$, $\mathbf{x}_3$

                \item \textbf{Call} trained matrix-learning model 1 to calculate $K_0'$, $K_1'$, $K_2'$, $K_3'$ from $K_0$, $K_1$, $K_2$, $K_3$
		
			\item \textbf{Call} QWENDY to calculate $A_1$ from $K_0'$, $K_1'$, $K_2'$, $K_3'$, $\mathbf{x}_0$, $\mathbf{x}_1$, $\mathbf{x}_2$, $\mathbf{x}_3$

			\item \textbf{Call} trained matrix-learning model 2 to calculate $A_2$ from $A_1$, $K_0$, $K_1$, $K_2$, $K_3$

                \item \textbf{Output}: inferred GRN $A_2$
			
		}
	\end{enumerate}
\end{algorithm}

For those two matrix-learning models in Algorithms~\ref{algt},\ref{alge}, we can adopt the approach of TRENDY to construct a transformer structure with three sections: (1) Pre-process the inputs; (2) Use transformer encoder layers to learn high-dimensional representations of the inputs; (3) Construct outputs from the high-dimensional representations. 

For training the transformer encoder layers in this structure, we present two approaches. One approach is to train new transformer encoder layers from scratch, the same as TRENDY. This structure is named ``TE'', meaning ``transformer-enhanced''. 

The other approach is to integrate an LLM and fine-tune it. In general, an LLM has three sections: (1) Convert text to vector representations; (2) Use transformer (encoder layers, decoder layers, or both) to learn contextual relationships; (3) Convert model outputs back to natural language \cite{zhao2023survey}. We can choose an LLM with pre-trained transformer encoder layers in the second section, and use them to replace the transformer encoder layers in the TE structure. In our practice, we use the encoder layers of the RoBERTa-large model \cite{liu2019roberta}, which have about 300 million parameters. 

The pre-trained parameters will be kept frozen while we adopt a parameter-efficient fine-tuning method: LoRA \cite{hu2021lora}. Lower-rank matrices will be trained and added to the frozen encoding layers to obtain a new encoder in a cost-efficient way. The new structure with RoBERTa is named ``LE'', representing ``LLM-enhanced''.

We name Algorithms~\ref{algt},\ref{alge} with TE structure as TEQWENDY, and Algorithms~\ref{algt},\ref{alge} with LE structure as LEQWENDY. See the technical details section for details of TE and LE structures. TEQWENDY has 4.7 million trainable parameters. LEQWENDY has 4.6 million trainable parameters, with 300 million non-trainable (frozen) parameters.

TEQWENDY and LEQWENDY are trained on synthetic data generated by \cite{pinna2010knockouts,papili2018sincerities,wang2024gene,tian2025trendy}
\begin{equation}
\mathrm{d}X_j(t)=V\left\{\beta \prod_{i=1}^n \left[1+(A_\Z)_{i,j}\frac{X_i(t)}{X_i(t)+1}\right]-\theta X_j(t)\right\}\mathrm{d} t +\sigma X_j(t)\mathrm{d}W_j(t),
\label{eqnl}
\end{equation}
where $X_i(t)$ is the level of gene $i$ at time $t$, $W_j(t)$ is a standard Brownian motion, and $V=30$, $\beta=1$, $\theta=0.2$, $\sigma=0.1$ . There are $10^5$ training samples, each with the expression levels of 10 genes for 100 cells, measured at four time points: 0.0, 0.1, 0.2, 0.3. 

QWENDY requires that $K_0$, $K_1$, $K_2$, $K_3$ are symmetric and positive definite. Although the input $K_i$ and the target $K_i^*$ are naturally symmetric and positive definite, the learned output $K_i'$ in Step (3) of Algorithms~\ref{algt},\ref{alge} might not be positive definite, or even symmetric. Therefore, in the implementation of QWENDY, we add two extra steps to adjust the input covariance matrices:

\noindent \textbf{1.} If $K_i$ is asymmetric, replace $K_i$ by $(K_i+K_i^\T)/2$.

\noindent \textbf{2.} If $K_i$ is not positive definite, in the eigenvalue decomposition $K_i=O\Lambda O^\T$, where $O$ is orthonormal, and $\Lambda$ is diagonal, replace negative values of $\Lambda$ by small positive values to obtain $\Lambda'$, and replace $K_i$ by $O\Lambda' O^\T$.

\subsection{Performance}

We also test the performance of TEQWENDY and LEQWENDY on the same data sets. See Table~\ref{tableapp} for the scores. Compared with QWENDY, TEQWENDY performs better on synthetic data, but worse on experimental data. The overall performance of TEQWENDY is worse than QWENDY. Therefore, TEQWENDY might overfit on synthetic data, which does not work on experimental data. The performance of LEQWENDY is worse than QWENDY on both synthetic and experimental data, meaning that LLMs trained with natural language inputs might not directly help with the numerical task in GRN inference, even after task-specific fine-tuning.

Training TEQWENDY and LEQWENDY on data generated by Eq.~\ref{eqnl} does not necessarily increase their performance on experimental data sets. This means that Eq.~\ref{eqnl} might not faithfully reflect the gene expression dynamics. To better integrate deep learning techniques, we need better gene expression data generators. Another possibility is that the training data are only from time 0.0 - 0.3. We could use data from later time points to increase the robustness of training.

\begin{table}[ht]
    \centering
    \begin{tabular}{ccccc}
        \toprule
         & &QWENDY & \begin{tabular}[c]{@{}l@{}}LEQ-\\ WENDY\end{tabular} & \begin{tabular}[c]{@{}l@{}}TEQ-\\ WENDY\end{tabular} \\
        \midrule
         &      &    &     &         \\
      \multirow{2}{*}{SINC}   & AUROC  & 0.5107  & 0.4990 & 0.5932 \\
         &  AUPRC  &0.5537  & 0.5183 & 0.6014 \\
         &      &     &    &         \\
      \multirow{2}{*}{DREAM4}   &   AUROC &0.4987  & 0.5164 & 0.5372 \\
         &  AUPRC  & 0.1844 & 0.1823 & 0.2203 \\
         &      &  &       &         \\
         \begin{tabular}[c]{@{}c@{}}Synthetic\\ total\end{tabular} &  & 1.7475   &  1.7160    &  1.9521           \\
         &      &     &    &         \\
       \multirow{2}{*}{THP-1}  &   AUROC  &0.5524 & 0.5543 & 0.5415 \\
         &   AUPRC & 0.4294 & 0.3632 & 0.3801 \\
         &      &  &       &         \\
       \multirow{2}{*}{hESC}  & AUROC  & 0.6019  & 0.5905 & 0.4815 \\
         &   AUPRC & 0.0435 & 0.0367 & 0.0317 \\
         &      &  &       &         \\
       \multirow{2}{*}{mESC}  & AUROC  & 0.5230  & 0.4793 & 0.2760 \\
         &   AUPRC & 0.0507 & 0.0432 & 0.0284 \\
         &      &  &       &         \\
         \begin{tabular}[c]{@{}c@{}}Experimental\\ total\end{tabular} & &  2.2009   &  2.0672    &   1.7392          \\
         &      &  &       &         \\
      \begin{tabular}[c]{@{}c@{}}Overall\\ total\end{tabular}   & & 3.9484    & 3.7832 & 3.6913 \\
         &      &   &      &         \\
       \begin{tabular}[c]{@{}c@{}}Overall rank\\in all 19\end{tabular}  &   &   5th   &  11th    & 15th     \\
        \bottomrule
    \end{tabular}
    \caption{AUROC and AUPRC scores of QWENDY, LEQWENDY, and TEQWENDY on four data sets. The overall rank is for 19 methods: 16 previously tested methods, QWENDY, LEQWENDY, and TEQWENDY}
    \label{tableapp}
\end{table}

\subsection{Technical details}

For all models in TEQWENDY and LEQWENDY, the loss function is mean squared error; the optimizer is Adam with learning rate $0.001$; the number of training epochs is $100$. After each epoch, we evaluate the model performance on a validation set with 1000 samples. The training stops early and rolls back to the best status if there is no improvement for consecutive $10$ epochs.

See Algorithm~\ref{algt1} for the structure of the first half of TEQWENDY. The inputs are four covariance matrices $K_0$, $K_1$, $K_2$, $K_3$. The targets are four revised covariance matrices $K_0^*$, $K_1^*$, $K_2^*$, $K_3^*$. Notice that the first input matrix ($K_0$) is not processed through these layers, since the target $K_0^*$ equals $K_0$. Besides, $K_1$, $K_2$, $K_3$ are processed separately.

\begin{algorithm}[!htbp]
	\caption{Structure of TEQWENDY method, first half. The shape of data after each layer is in the brackets.}
	\label{algt1}
	\ \\
	\begin{enumerate}
		{	\item \textbf{Input}: four covariance matrices ($4$ groups of $n\times n$)
			
			\item Linear embedding layer with dimension $1$ to $d=64$ ($4$ groups of $n\times n \times d$)

                \item ReLU activation function ($4$ groups of $n\times n \times d$)

                \item Linear embedding layer with dimension $d$ to $d$ ($4$ groups of $n\times n \times d$)
		
			\item 2-D positional encoding layer ($4$ groups of $n\times n \times d$)

			\item Flattening and concatenation ($4$ groups of $(n^2) \times d$)

                \item $7$ layers of transformer encoder with $4$ heads, dimension $d$, feedforward dimension $4d$, dropout rate $0.1$ ($4$ groups of $(n^2) \times d$)

                \item Linear embedding layer with dimension $d$ to $d$ ($4$ groups of $(n^2) \times d$)

                \item LeakyReLU activation function with $\alpha=0.1$ ($4$ groups of $n\times n \times d$)

                \item Linear embedding layer with dimension $d$ to $1$ ($4$ groups of $(n^2) \times 1$)

                \item \textbf{Output}: reshaping into four matrices ($4$ groups of $n \times n$)

		}
	\end{enumerate}
\end{algorithm}

The 2-D positional encoding layer incorporates spatial information of matrix to the input. It generates an $n\times n \times d$ array $\text{PE}$ and adds it to the embedded input:
For $x$ and $y$ in $1,2,\ldots,n$ and $j$ in $1,\ldots,d/4$,
\begin{equation*}
    \begin{split}
         &\text{PE}[x,y,2j-1]=\cos[(x-1)\times 10^{-16(j-1)/d}],\\
&\text{PE}[x,y,2j]=\sin[(x-1)\times 10^{-16(j-1)/d}],\\
& \text{PE}[x,y,2j-1+d/2]=\cos[(y-1)\times 10^{-16(j-1)/d}],\\
& \text{PE}[x,y,2j+d/2]=\sin[(y-1)\times 10^{-16(j-1)/d}].
    \end{split}
\end{equation*}

See Algorithm~\ref{algt2} for the structure of the second half of TEQWENDY. After obtaining the outputs $K_0'$, $K_1'$, $K_2'$, $K_3'$ from $K_0$, $K_1$, $K_2$, $K_3$ by the first half of TEQWENDY, call the QWENDY method to calculate the inferred GRN $A_1$ from $K_0'$, $K_1'$, $K_2'$, $K_3'$. The inputs of the second half of TEQWENDY are four covariance matrices $K_0$, $K_1$, $K_2$, $K_3$, and the inferred GRN $A_1$. The target is the ground truth GRN $A_\Z$. 

\begin{algorithm}[!htbp]
	\caption{Structure of TEQWENDY method, second half. The shape of data after each layer is in the brackets.}
	\label{algt2}
	\ \\
	\begin{enumerate}
		{	\item \textbf{Input}: four covariance matrices and one inferred GRN ($5$ groups of $n\times n$)
			
			\item Linear embedding layer with dimension $1$ to $d=64$ ($5$ groups of $n\times n \times d$)

                \item Segment embedding layer ($5$ groups of $n\times n \times d$)
		
			\item 2-D positional encoding layer ($5$ groups of $n\times n \times d$)

			\item Flattening and concatenation ($(n^2) \times (5d)$)

                \item $3$ layers of transformer encoder with $4$ heads, dimension $5d$, feedforward dimension $20d$, dropout rate $0.1$ ($(n^2) \times (5d)$)

                \item Linear embedding layer with dimension $5d$ to $1$ ($(n^2) \times 1$)

                \item \textbf{Output}: reshaping into one matrix ($n \times n$)
			
		}
	\end{enumerate}
\end{algorithm}

The segment embedding layer generates different trainable $d$-dimensional vectors for all five inputs. Then each vector is copied into dimension $n\times n \times d$, and added to the embedded inputs. This layer incorporates the source of inputs. The final input of the transformer encoder layers is a matrix, with shape $k\times D$, where $k$ is the total number of input values, and $D$ is the representation dimension. For each location in $1,2,\ldots,k$, the segment embedding marks which input matrix it is from, and the position encoding marks which position in the matrix it is from. These two layers solve the problem that the inputs are multiple matrices, but the inputs of transformer are representations of a 1-D sequence.

See Algorithm~\ref{algl1} for the structure of the first half of LEQWENDY. The inputs are four covariance matrices $K_0$, $K_1$, $K_2$, $K_3$. The targets are four revised covariance matrices $K_0^*$, $K_1^*$, $K_2^*$, $K_3^*$. Notice that the first input matrix ($K_0$) is not processed through these layers, since the target $K_0^*$ equals $K_0$. Besides, $K_1$, $K_2$, $K_3$ are processed separately.

\begin{algorithm}[!htbp]
	\caption{Structure of LEQWENDY method, first half. The shape of data after each layer is in the brackets.}
	\label{algl1}
	\ \\
	\begin{enumerate}
		{	\item \textbf{Input}: four covariance matrices ($4$ groups of $n\times n$)
			
			\item Linear embedding layer with dimension $1$ to $d=256$ ($4$ groups of $n\times n \times d$)

                \item ReLU activation function ($4$ groups of $n\times n \times d$)

                \item Linear embedding layer with dimension $d$ to $d$ ($4$ groups of $n\times n \times d$)

                \item Segment embedding layer ($4$ groups of $n\times n \times d$)
		
			\item 2-D positional encoding layer ($4$ groups of $n\times n \times d$)

			\item Flattening and concatenation ($(n^2) \times (4d)$)

                \item Transformer encoder part of the RoBERTa-large model, frozen: $24$ layers of transformer encoder with $16$ heads, dimension $4d$, feedforward dimension $16d$, dropout rate $0.1$; 
                
                Trainable LoRA layers with rank $r=8$ and $\text{LoRA-}\alpha=16$, added to each transformer encoder layer ($(n^2) \times (4d)$)

                \item Four different linear embedding layers with dimension $4d$ to $2d$ ($4$ groups of $(n^2) \times (2d)$)

                \item LeakyReLU activation function with $\alpha=0.1$ ($4$ groups of $(n^2) \times (2d)$)

                \item Four different linear embedding layers with dimension $2d$ to $1$ ($4$ groups of $(n^2) \times 1$)

                \item \textbf{Output}: reshaping into four matrices ($4$ groups of $n \times n$)

		}
	\end{enumerate}
\end{algorithm}

For each large pre-trained weight matrix $W$ with size $p\times q$, LoRA freezes $W$ and replace it by $W+\Delta W$. Here $\Delta W=(\text{LoRA-}\alpha/r)AB$, where the trainable $A$ has size $p \times r$, and the trainable $B$ has size $r\times q$. The total number of trainable parameters decreases from $pq$ to $(p+q)r$, since $r\ll p,q$. The scaling factor $\text{LoRA-}\alpha$ controls the update rate.

See Algorithm~\ref{algl2} for the structure of the second half of LEQWENDY. After obtaining the outputs $K_0'$, $K_1'$, $K_2'$, $K_3'$ from $K_0$, $K_1$, $K_2$, $K_3$ by the first half of LEQWENDY, call the QWENDY method to calculate the inferred GRN $A_1$ from $K_0'$, $K_1'$, $K_2'$, $K_3'$. The inputs of the second half of LEQWENDY are four covariance matrices $K_0$, $K_1$, $K_2$, $K_3$, and the inferred GRN $A_1$. The target is the ground truth GRN $A_\Z$. Since the transformer encoder part of the RoBERTa-large model has a fixed dimension 1024, we need to apply $d_1=192$ for each covariance matrix input, and $d_2=256$ for the GRN input, so that the total dimension is $d=4d_1+d_2=1024$.

\begin{algorithm}[!htbp]
	\caption{Structure of LEQWENDY method, second half. The shape of data after each layer is in the brackets.}
	\label{algl2}
	\ \\
	\begin{enumerate}
		{	\item \textbf{Input}: four covariance matrices and one inferred GRN ($5$ groups of $n\times n$)
			
			\item Linear embedding layer with dimension $1$ to $d_1=192$ or $d_2=256$ ($4$ groups of $n\times n \times d_1$ and $1$ group of $n\times n \times d_2$)

                \item Segment embedding layer ($4$ groups of $n\times n \times d_1$ and $1$ group of $n\times n \times d_2$)
		
			\item 2-D positional encoding layer ($4$ groups of $n\times n \times d_1$ and $1$ group of $n\times n \times d_2$)

			\item Flattening and concatenation ($(n^2) \times (d=4d_1+d_2=1024)$)

                \item Transformer encoder part of the RoBERTa-large model, frozen: $24$ layers of transformer encoder with $16$ heads, dimension $4d$, feedforward dimension $16d$, dropout rate $0.1$; 
                
                Trainable LoRA layers with rank $r=16$ and $\text{LoRA-}\alpha=32$, added to each transformer encoder layer ($(n^2) \times d$)

                \item Linear embedding layer with dimension $d$ to $d$ ($(n^2) \times d$)

                \item LeakyReLU activation function with $\alpha=0.1$ ($(n^2) \times d$)

                \item Linear embedding layer with dimension $d$ to $1$ ($(n^2) \times 1$)

                \item \textbf{Output}: reshaping into one matrix ($n \times n$)
			
		}
	\end{enumerate}
\end{algorithm}

\bibliographystyle{unsrt}
\bibliography{qwendy}
\end{document}